\documentclass[smallextended,draft,nospthms,natbib]{svjour3}
\usepackage[margin=1in]{geometry}
\usepackage{amssymb,amsthm,amsmath}
\usepackage{enumerate}
\usepackage[all]{xy}
\usepackage{setspace}

\usepackage{float}
\usepackage{natbib}

\newtheorem{thm}{Theorem}

\newtheorem{prop}{Proposition}

\theoremstyle{definition}

\theoremstyle{definition}

\providecommand{\norm}[1]{\lVert#1\rVert}
\providecommand{\abs}[1]{\lvert#1\rvert}
\providecommand{\inner}[2]{\langle#1,#2\rangle}

\date{}

\title{Localizable Particles in the Classical Limit of Quantum Field Theory}
\author{Benjamin H.~Feintzeig \and Jonah Librande \and Rory Soiffer}
\institute{{Benjamin H.~Feintzeig \at Department of Philosophy, University of Washington \and Jonah Librande \at Department of Mathematics, University of Washington \and Rory Soiffer \at Department of Computer Science, University of Waterloo}}
\date{}

\begin{document}

\maketitle

\begin{abstract}
A number of arguments purport to show that quantum field theory cannot be given an interpretation in terms of localizable particles.  We show, in light of such arguments, that the classical $\hbar\to 0$ limit can aid our understanding of the particle content of quantum field theories.  In particular, we demonstrate that for the massive Klein-Gordon field, the classical limits of number operators can be understood to encode local information about particles in the corresponding classical field theory. 
\end{abstract}

\keywords{quantum field theory \and particle interpretation \and classical limit \and deformation quantization}

\section{Introduction}
\label{sec:intro}

Relativistic quantum field theory underlies the modern discipline of particle physics.  Practitioners use the theory to conceptualize interactions between particles and make quantitative predictions about scattering experiments.  Yet a number of arguments purport to show that various features of our particle concept are incompatible with the constraints of relativistic quantum physics.\footnote{For a nice comprehensive review of a broader collection of issues for particle interpretations than will be discussed in this paper, see \citet{Fr20}.  \citet{Ba16} also contains an introduction to issues with particle interpretations in the context of philosophy of quantum field theory more generally.}  Building on results of \citet{Ma96}, \citet{HaCl02} argue that in relativistic quantum theory, particles cannot be localized in spatial regions.  The conclusions of such arguments leave interpreters of relativistic quantum field theory with a puzzle.  How can an underlying theory that does not allow for localized particles support descriptive and explanatory practices that appear to involve localized particles?

Previous investigations have focused on the issue of recovering the phenomenology of particle physics from relativistic quantum field theory.  For example, \citet{Bu95} provides a way of recovering scattering theory at asymptotic times.\footnote{Note that in the limit of asymptotic times taken for scattering theory, one only has the ability to describe momentum states and one loses the notion of an exactly localizable particle. In the mathematical physics literature, analyses of particles proceed via the technical notion of ``almost local" particle observables, which are used to analyze the localization structure of the one-particle subspace.  For mathematical development, see, e.g., \citet{BuFr82,Ha92,BuPoSt91,Bu95}.  For philosophical discussion of ``almost local" particle observables, see \citet{HaCl02,ArSt13,Va15}.  In this paper, we also deal with approximate localization, but in a somewhat different sense.  Moreover, we analyze localization properties of number operators in the full field theory without restricting attention to single particle states.}  In contrast, the goal of this paper is to make a small contribution toward our understanding of the theoretical role of particles in quantum field theory.  We aim to make precise a sense in which a theoretical description of particles \emph{emerges} from quantum field theory through the behavior of number operators.

In this vein, we follow \citet{Wa01}, who argues for the emergence of particles in terms of the approximate localization, in a certain sense, of structures in quantum field theory.\footnote{For earlier mathematical work on localization, see \citet{Kn61} and \citet{Li63}.  For earlier philosophical work on localization, see \citet{Sa92,Sa95} and \citet{Re95a,Re95b}.} Similarly, recent work by \citet{PaPy19} employs the non-relativistic limit as an approximation to analyze the localizability of particles.  Our investigation complements this work by instead analyzing the classical $\hbar\to 0$ limit of quantum field theory.  We take this approach because there are existing tools for formulating the $\hbar\to 0$ limit in the C*-algebraic framework for quantum theory with full mathematical rigor \citep[see][]{La17}.  Moreover, \citet{La13} has already initiated the use of these tools to analyze emergent behavior in quantum theories.  We believe many approaches are helpful for understanding particles in quantum field theory, and so we will pursue an analysis through the classical limit without trying to rule out other avenues.  We hope the positive outcome of our analysis of the classical limit speaks in favor of the approach taken here, but we do not believe it speaks against other approaches to understanding particle-like behavior.

We wish to distinguish the results in this paper from a number of other recent approaches to understanding particles in quantum field theories.  First, some recent work on pilot wave (i.e.,  Bohmian-type) quantum field theories allows one to understand the content of those theories in terms of particles.  Early work in this direction can be traced to \citet[][p. 174-7]{Be87}; for an overview of recent progress, see \citet{St10,Struyve11}.  Second, some recent work has led to the development of a ``dissipative" approach to quantum field theory, using tools from non-equilibrium thermodynamics \citep{OlOt21}, which those authors argue can support a particle interpretation. Both of these approaches involve substantial modification of the traditional framework for quantum field theory, as is clear in their use of different dynamical laws.  In contrast, our task in this paper is to search for possible routes to understanding particle-like behavior within the standard formulation of quantum field theory without modifying the theory itself.  We make no judgments concerning these distinct approaches to particles; we only claim that they are not relevant to the question we treat.  Lastly, we mention the recent work of \citet{Bi18}, who argues that the standard textbook description of modes in quantum field theory representing an ontology of particles is unwarranted.  While our starting point of standard quantum field theory agrees with that used by \citet{Bi18}, we employ very different methods to yield a particle interpretation.  Instead of interpreting the formalism of quantum field theory \emph{directly} in terms of particles, we consider only \emph{indirectly} how particle-like behavior may emerge from quantum field theory in the $\hbar\to 0$ limit.

Our goal in this paper is to argue that the classical limit helps us understand particle content in quantum field theory in terms of classical field theoretic quantities.  Our strategy of interpreting quantum field theories in terms of relationships between theories at different scales and emergent structures in some sense follows the approach advocated by \citet{Wa06} and \citet{Wi17}.\footnote{These authors at times suggest their interpretive approach is somehow incompatible with or not conducive to the mathematical analysis provided by algebraic quantum field theory.  By employing what we take to be a similar interpretive approach by using the tools of algebraic quantum field theory, we believe we demonstrate the compatibility of these approaches in this paper.}  Both of those authors emphasize the importance of an often overlooked interpretive task.  Philosophers often aim to answer the question: ``if this theory provided a true description of the world in all respects, what would the world be like?"  But Wallace and Williams claim it is important to also consider the distinct question ``given that this theory provides an approximately true description of our world, what is our world approximately like?" \citep[][p. 210]{Wi17}.  In this paper, we only aim to contribute to the latter task, and only for certain approximative regimes.  It is only with regard to this question that we are interpreting quantum field theory at all as opposed to the classical field theories that we deal with more directly.  As such we make no claims about what one might call the ``fundamental ontology" of quantum field theory, but we still take our conclusions to be important to the interpretation and understanding of that theory.

The paper is structured as follows.  In \S\ref{sec:back}, we review the arguments against localized particle interpretations that we engage with in this paper.  In \S\ref{sec:class}, we summarize recent technical work concerning classical limits of number operators in quantum field theory within a C*-algebraic formulation of the classical limit.  In \S\ref{sec:part}, we use these results to clarify the theoretical status of localized particles in quantum field theory by providing two possible interpretations of approximately localizable particles.  We conclude with some discussion in \S\ref{sec:con}.

\section{The Case Against Localizable Particles}
\label{sec:back}

Theorems due to \citet{Ma96} and \citet{HaCl02} aim to show that relativistic quantum theories cannot allow for localizable particles.\footnote{See also the related results of \citet{He98,He98b}.}  We begin with Malament's no-go result concerning particle positions, and then present Halvorson and Clifton's refinement for local number operators.

Malament targets the existence of a position operator in relativistic quantum theory.  If there were a position operator, then for any foliation of Minkowski spacetime $M$ into spacelike hyperplanes, there would be a projection-valued measure on bounded open subsets of those hyperplanes serving as the position operator's spectral decomposition.  Suppose we are given such a foliation; call a bounded open subset of one of the hyperplanes a \textit{spatial set}.  A \textit{localization system} is defined as a triple $(\mathcal{H},\Delta\mapsto P_\Delta,a\mapsto U(a))$, where $\mathcal{H}$ is a Hilbert space, each spatial set $\Delta$ is assigned a projection $P_\Delta$ on $\mathcal{H}$, and $a\mapsto U(a)$ is a strongly continuous unitary representation of the translation group of $M$.  We interpret each projection $P_\Delta$ as representing the proposition that the particle is located within $\Delta$.

Malament considers the following constraints on a localization system:
\begin{enumerate}

    \item \textit{Translation Covariance}: for all spatial sets $\Delta$ and all vectors $a$ in $M$,
    \[U(a)P_\Delta U(a)^* = P_{\Delta+a}.\]
    
    \item \textit{Energy Condition}: for each future-directed unit timelike vector $a$ in $M$, the unique self-adjoint generator\footnote{The existence of such a generator is guaranteed by Stone's theorem \citep[see][p. 264]{ReSi80}.} of the one-parameter unitary family $t\in\mathbb{R}\mapsto U(ta)$ has a spectrum bounded from below.
    
    \item \textit{Microausality}: if $\Delta_1,\Delta_2$ are spacelike separated spatial sets, then
    \[P_{\Delta_1}P_{\Delta_2} = P_{\Delta_2}P_{\Delta_1}.\]
    
    \item \textit{Localizability}: if $\Delta_1,\Delta_2$ are disjoint spatial sets in the same hyperplane, then
    \[P_{\Delta_1}P_{\Delta_2} = P_{\Delta_2}P_{\Delta_1} = 0.\]

\end{enumerate}
Translation Covariance allows us to understand the unitary representation of the translation group as providing a link between the propositions associated with translated spatial sets.  The Energy Condition guarantees that one cannot extract an infinite amount of energy from the system.  The Microcausality condition enforces the relativistic constraint of no ``act-outcome" correlations between spacelike separated events.  And the Localizability condition ensures that a particle cannot be found in two disjoint spatial sets at the same time.

With these constraints, Malament proves the following result:
\begin{thm}[\citeauthor{Ma96}]
Suppose a localization system $(\mathcal{H},\Delta\mapsto P_\Delta,a\mapsto U(a))$ satisfies conditions (1)-(4).  Then $P_\Delta = 0$ for all spatial sets $\Delta$.
\end{thm}
\noindent If $P_\Delta = 0$ for all spatial sets, then the probability of finding the particle in any spatial region is zero.  In this case, such a structure cannot be used to represent a localizable particle position.  Hence, the theorem may be interpreted as a ``no-go" result, showing that no particle position operator can exist in a relativistic quantum theory.

Those who are trained in modern particle physics may be skeptical of the upshot of Malament's theorem for relativistic quantum \emph{field} theory, where we do not typically employ particle position operators.  In fact, Malament's own interpretation is that his result pushes one towards field theories rather than particle theories in the relativistic setting.  In standard formulations of quantum field theory, however, one employs particle number operators rather than position operators.  One might think that this is enough to avoid the consequences of Malament's theorem if we could understand such number operators as being associated with the number of particles in a spatial set.  But \citet{HaCl02} prove a result analogous to Malament's theorem demonstrating that localizable number operators are not compatible with certain constraints of relativistic quantum theory.

\citet{HaCl02} consider the same setting as above, except that they allow spatial sets to be arbitrary (not necessarily open) bounded subsets.\footnote{This is necessary so that the Number Conservation condition below is non-trivial as there are no countable disjoint coverings of a hyperplane by open sets.}  They define a \textit{system of local number operators} as a triple $(\mathcal{H},\Delta\mapsto N_\Delta,a\mapsto U(a))$, where $\mathcal{H}$ is a Hilbert space, each spatial set $\Delta$ is assigned an operator $N_\Delta$ on $\mathcal{H}$ with eigenvalues $\{0,1,2,...\}$, and $a\mapsto U(a)$ is a strongly continuous unitary representation of the translation group of $M$.  We now interpret each operator $N_\Delta$ as representing the number of particles in the spatial set $\Delta$.

Halvorson and Clifton consider the constraints of Translation Covariance, the Energy Condition, and Microcausality, which carry over in the straightforward way for the operators $N_{\Delta}$.  They append to this list the conditions:
\begin{enumerate}\setcounter{enumi}{4}
\item \textit{Additivity}: if $\Delta_1,\Delta_2$ are disjoint spatial sets in the same hyperplane, then
\[N_{\Delta_1} + N_{\Delta_2} = N_{\Delta_1\cup \Delta_2}.\]
\item \textit{Number Conservation}: if $\{\Delta_n\}_{n\in\mathbb{N}}$ is a countable disjoint covering of a hyperplane consisting of spatial sets, then
\[\overline{N}:=\sum_n N_{\Delta_n}\]
converges to a self-adjoint operator on $\mathcal{H}$ satisfying
\[U(a)\overline{N}U(a)^* = \overline{N}\]
for any timelike vector $a$ in $M$.

\end{enumerate}
The Additivity condition tells us that for any state, the expectation value of the number of particles in the union of two disjoint spatial regions in the same hyperplane is the sum of the expectation value for the number of particles in each of the two regions.  The Number Conservation condition guarantees the existence of a total number operator $\overline{N}$ representing the number of particles contained in all spatial regions, whose expectation value is constant in time.  Halvorson and Clifton assert Number Conservation is a reasonable condition for free field theories, even if it may not hold in interacting field theories.

With these constraints, Halvorson and Clifton prove the following result.
\begin{thm}[\citeauthor{HaCl02}]
\label{thm:locnum}
Suppose a system of local number operators $(\mathcal{H},\Delta\mapsto N_\Delta,a\mapsto U(a))$ satisfies conditions (1)-(3) and (5)-(6).  Then $N_\Delta = 0$ for all spatial sets $\Delta$.
\end{thm}
\noindent The conclusion that $N_\Delta=0$ for all spatial sets implies that we will not find particles in any spatial region.  Again, such a structure is incapable of representing particles.  So this serves as a ``no-go" result for localizable particles in free field theories.

\section{The Classical Limit}
\label{sec:class}

While the result of Halvorson and Clifton described in \S\ref{sec:back} shows that number operators in quantum field theories cannot be associated with local regions while satisfying their constraints, it is well known that free quantum field theories allow for the definition of number operators associated with the entire system.  In other words, all parties agree that number operators exist; what is at issue in the ``no-go" results just reviewed is whether they can be associated with local regions.  In this section, we will review recent work on the classical limits of number operators in free quantum field theories, using the free massive Klein-Gordon field as a concrete illustration.  Then we will go on to use the classical limit to aid in understanding the localization of number operators next.  In \S\ref{sec:QFT} we review the construction of number operators in quantum field theory, and in \S\ref{sec:quant} we review an analysis of their classical limits.

\subsection{Quantum Field Theory}
\label{sec:QFT}
We aim to construct number operators in quantum field theories by first specifying an abstract C*-algebra of bounded quantities, and then considering Hilbert space representations of this algebra on a standard Fock space.  One can obtain field operators as certain limits of bounded quantities, and then construct number operators from the fields.  We will emphasize that the tool of a complex structure used to construct the representing Hilbert space also determines the form of the number operator.

One can construct free bosonic quantum field theories with the so called \textit{Weyl (or CCR) algebra}.  In a free field theory on Minkowski spacetime, one starts with a symplectic vector space $(E,\sigma)$ of test functions whose dual space is the space of (possibly distributional) solutions to the field equations, or equivalently, initial data on a Cauchy surface $\Sigma\cong\mathbb{R}^3$.  The kinematical algebra of bounded physical quantities can then be specified by the Weyl algebra $\mathcal{W}(E,\hbar\sigma)$, the smallest C*-algebra generated by the linearly independent elements $W_\hbar(F),W_\hbar(G)$ for each $F,G\in E$ with operations
\begin{align}
    W_\hbar(F) W_\hbar(G) &:= e^{-\frac{i\hbar}{2}\sigma(f,g)}W_\hbar(F+G)\\
    W_\hbar(F)^* &:= W_\hbar(-F) \nonumber
\end{align}
and the minimal regular norm \citep[see][]{MaSiTeVe74,Pe90}.

For example, for a real scalar field $\varphi: M\to \mathbb{R}$ satisfying the Klein-Gordon equation\footnote{The Klein-Gordon equation is often presented with $m^2$ replaced by $m^2/\hbar^2$.  Since we are considering a quantum field theory whose classical $\hbar\to 0$ limit is the classical Klein-Gordon field with finite mass, our setup builds commutation relations for the quantum theory from Eq. (\ref{eq:KG}) with no factors of $\hbar$.  The version of the Klein-Gordon equation we use is truly classical in the sense that it does not depend on $\hbar$.  However, one can also interpret this setup as applying to the Klein-Gordon equation with the factors of $\hbar$ included (with $m^2$ replaced by $m^2/\hbar^2$) by understanding the limit to encode simultaneous rescalings of $m$ and $\hbar$ in such a way that $m^2/\hbar^2$ remains constant.  This can happen, for example, if one considers simultaneous unit changes for mass, distance, and time that hold fixed the value of the speed of light $c$.  As an aside, we mention that one can also investigate the alternative limit $m^2/\hbar^2\to 0$ by employing renormalization techniques \citep{BuVe95,BuVe98}.}
\begin{align}
\label{eq:KG}
    \frac{\partial^2\varphi}{\partial t^2} - \nabla^2\varphi = -m^2\varphi,
\end{align}
we set\footnote{$C_c^\infty(\Sigma)$ denotes the real vector space of smooth, compactly supported real-valued functions on $\Sigma$.} $E = C_c^\infty(\Sigma)\oplus C_c^\infty(\Sigma)$ and define $\sigma$ for all $(f_1,g_1),(f_2,g_2)\in E$ by
\begin{align}
    \sigma((f_1,g_1),(f_2,g_2)) := \int_{\mathbb{R}^3} f_1g_2-f_2g_1.
\end{align}
We understand a (possibly distributional) solution to the Klein-Gordon equation $(\pi,\varphi)$ with $\pi = \frac{\partial \varphi}{\partial t}$ to be an element of $E'$ with action $(\pi,\varphi)[f,g] := \int_\Sigma \pi f + \varphi g$ for $(f,g)\in E$.

One can construct particle number operators in the quantum theory by focusing on Fock space representations of the Weyl algebra.  A Fock space representation arises from a choice of timelike spacetime symmetry group acting on $E$, which determines a \textit{complex structure} on $E$, i.e. a linear map $J:E\to E$ satisfying \citep[see][]{ClHa01,Ka79}
\begin{enumerate}[(i)]
\item $\sigma(JF,JG) = \sigma(F,G)$;
\item $\sigma(F,JF)\geq 0$;
\item and $J^2 = -I$
\end{enumerate}
for all $F,G\in E$.

For example, the inertial timelike symmetries of Minkowski spacetime determine the \emph{Minkowski complex structure} $J_M$ for the Klein-Gordon field, defined as follows.  First, define the differential operator $\mu_M:C_c^\infty(\Sigma)\to C_c^\infty(\Sigma)$ by
\begin{align}
    \mu_M := (m^2-\nabla^2)^{1/2}.
\end{align}
Then define $J_M: E\to E$ by
\begin{align}
J_M(f,g) := (-\mu_M^{-1}g,\mu_Mf)    
\end{align}
for all $(f,g)\in E$.  This choice of $J_M$ is the unique complex structure that commutes with inertial time evolution of solutions given by $f\mapsto e^{-i\mu_Mt}f$ for $f\in C_c^\infty(\Sigma)$ and $t\in\mathbb{R}$.  But this is only one possible choice of complex structure.

In addition, the Lorentz boost symmetries of the Rindler wedge in Minkowski spacetime determine the \emph{Rindler complex structure} $J_R$ for the Klein-Gordon field, defined as follows.  Consider initial data for the right Rindler wedge on the surface $R = \{(x,y,z)\in \Sigma\ |\ x>0\}$.  Consider the space $E(R) = C_c^\infty(R)\oplus C_c^\infty(R)$ of test functions with support on $R$.  Define the differential operator $\mu_R:C_c^\infty(R)\to C_c^\infty(R)$ by
\begin{align}
    \mu_R:= \bigg(e^{2x}\Big(m^2-\frac{\partial^2}{\partial y^2} - \frac{\partial^2}{\partial z^2}\Big) - \frac{\partial}{\partial x^2}\bigg)^{1/2}.
\end{align}
Then define $J_R:E(R)\to E(R)$ by
\begin{align}
    J_R(f,g) := (-\mu_R^{-1}g,\mu_Rf)
\end{align}
for all $(f,g)\in E(R)$. This choice of $J_R$ is the unique complex structure that commutes with time evolution of solutions in Rindler coordinates (Lorentz boosts) given by $f\mapsto e^{-i\mu_Rt}f$ for $f\in C_c^\infty(R)$ and $t\in\mathbb{R}$.  Thus, distinct complex structures arise from distinct choices of timelike symmetry groups, and as we now discuss, each of these complex structures defines a different number operator associated with a representation of the Weyl algebra.\footnote{For more detail on each of these complex structures and domain issues, see \citet{Ka85}.}

A complex structure allows us to construct a Hilbert space known as a Fock space, which carries the structures of interest.  We can define a Fock space by first noticing that a complex structure $J$ determines a complex inner product on $E$ defined by
\begin{align}
    \alpha_J(F,G) := \sigma(F,JG) + i\sigma(F,G)
\end{align}
for all $F,G\in E$.  The completion of $E$ with respect to this inner product is a complex Hilbert space $\mathcal{H}_J$.  The Fock space over $\mathcal{H}_J$ is then defined as
\begin{align}
    \mathcal{F}(\mathcal{H}_J) := \bigoplus_{n=0}^\infty \mathcal{S}\bigg(\bigotimes^n\mathcal{H}_J\bigg),
\end{align}
where $\mathcal{S}(\bigotimes^n\mathcal{H}_J)$ denotes the symmetric subspace of the $n$-fold tensor product of $\mathcal{H}_J$ with itself, and which for $n=0$ is defined as $\mathbb{C}$.  This Hilbert space carries a representation of the Weyl algebra unitarily equivalent to the GNS representation $\pi_\omega$ for the state $\omega$ defined by
\begin{align}
    \omega(W_\hbar(f)) := e^{-\frac{\hbar}{4}\alpha_J(F,F)}
\end{align}
for all $F\in E$.  The state $\omega$ is the unique vacuum state invariant under the timelike symmetries we began with.  Since the state $\omega$ is regular, the Fock space carries unbounded field operators $\Phi_\hbar(F)$, which are the self-adjoint generators of the one-parameter unitary families $t\in\mathbb{R}\mapsto \pi_\omega(W_\hbar(tF))$.  These field operators can be used to define the standard creation and annihilation operators and number operators
\begin{align}
\label{eq:numberdef}
a_\hbar^J(F) := \frac{1}{\sqrt{2}}\bigg(\Phi_\hbar(F) + i\Phi_\hbar(&JF)\bigg)\ \ \ \ \ \ \ \ \ \ (a_\hbar^J(F))^* := \frac{1}{\sqrt{2}}\bigg(\Phi_\hbar(F) -i\Phi_\hbar(JF)\bigg)\\
 &N_\hbar^J(F) := (a_\hbar^J(F))^*a_\hbar^J(F) \nonumber
\end{align}
for each $F\in E$ \citep[for more detail, see][]{ClHa01}.

There are multiple distinct Fock space representations of the Weyl algebra for the Klein-Gordon field.  For example, if one chooses the complex structure $J_M$ corresponding to the inertial timelike translation symmetries on Minkowski spacetime, then one arrives at the standard Minkowski vacuum, which we denote $\omega_M$, with its associated Fock space representation $(\pi_M,\mathcal{F}(\mathcal{H}_M))$.  But if one chooses the complex structure $J_R$ corresponding to the timelike Lorentz boost symmetries of the right Rindler wedge, then one arrives at the Rindler vacuum, which we denote $\omega_R$, with its associated Fock space representation $(\pi_R,\mathcal{F}(\mathcal{H}_R))$ \citep[see][]{Ka85}.  With respect to the Lorentz boost symmetries of the Rindler wedge, the Minkowski vacuum state appears as a thermal (KMS) state with finite temperature.\footnote{For more on the definition and interpretation of KMS states, see \citet{BrRo87,BrRo96,Ru11}.}  This is the celebrated Unruh effect \citep[see][]{Ka85,ArEaRu03,Ea11}, which is taken to imply that an observer accelerating through the Rindler wedge will observe a finite temperature, and hence particles, even in the Minkowski vacuum.  Thus, we have two families of number operators for further analysis.

\subsection{Strict Quantization and Number Operators}
\label{sec:quant}

Given the issues with particle localization in relativistic quantum field theory, we now seek a positive account of the particle-like content of such theories.  We will aim at an analysis of number operators that allows us to understand them as \emph{approximately} representing localizable particles.  The particular approximation we choose involves the classical $\hbar\to 0$ limit \citep[see also][]{Fe19}, although we make no claim that this is the only relevant approximation.  This section provides the relevant technical background and summarizes recent results of \citet{BrFeGaLiSo20} concerning the classical limits of number operators.  Those authors construct a quantization map that allows for the analysis of classical limits of unbounded operators without a choice of Hilbert space representation.  Applying this analysis to number operators yields integral expressions for the particle number content of the corresponding classical field theory, which we will analyze in the following sections.

One can formulate the classical limit using the mathematical framework of \textit{strict quantization} \citep[see][]{Ri89,Ri93,La98b,La17}.  A strict quantization is a family of C*-algebras $\{\mathfrak{A}_\hbar\}_{\hbar\in[0,1]}$ and a family of \emph{quantization maps} $\{\mathcal{Q}_\hbar: \mathcal{P}\to\mathfrak{A}_\hbar\}_{\hbar\in[0,1]}$ defined on some Poisson subalgebra $\mathcal{P}$ of $\mathfrak{A}_0$, which is required to be commutative.  The idea is that  $\mathfrak{A}_\hbar$ represents the collection of quantities in the quantum theory where Planck's constant takes on the value $\hbar\in(0,1]$, while $\mathfrak{A}_0$ represents the collection of quantities of the corresponding classical theory.  To appropriately capture the limiting behavior of the algebraic structure, a strict quantization is required to satisfy:
\begin{enumerate}[(i)]
\item $\lim_{\hbar\to 0}\norm{\frac{i}{\hbar}[\mathcal{Q}_\hbar(A),\mathcal{Q}_\hbar(B)] - \mathcal{Q}_\hbar(\{A,B\})}_\hbar = 0$;
\item $\lim_{\hbar\to 0}\norm{\mathcal{Q}_\hbar(A)\mathcal{Q}_\hbar(B) - \mathcal{Q}_\hbar(AB)}_\hbar = 0$; and
\item the map $\hbar\mapsto \norm{\mathcal{Q}_\hbar(A)}_\hbar$ is continuous
\end{enumerate}
for each $A,B\in\mathcal{P}$, where $\norm{\cdot}_\hbar$ is the C*-norm on $\mathfrak{A}_\hbar$.\footnote{Furthermore, it is typically required that $\mathcal{Q}_\hbar(\mathcal{P})$ is norm dense in $\mathfrak{A}_\hbar$, but this can always be achieved by restricting attention to an appropriate C*-subalgebra of the codomain.}  Given such a structure, we understand the classical limit of the family of quantities $\mathcal{Q}_\hbar(A)$ to be the classical quantity $A\in\mathcal{P}$.

In our example of the free Klein-Gordon field, we let $\mathfrak{A}_\hbar = \mathcal{W}(E,\hbar\sigma)$ and define the quantization maps as the linear extension of
\begin{align}
\label{eq:quantmap}
    \mathcal{Q}^{\alpha}_\hbar(W_0(F)) := e^{-\frac{\hbar}{4}\alpha(F,F)}W_\hbar(F),
\end{align}
for all $F\in E$, where $\alpha$ is any complex inner product on $E$.  It follows from results of \citet{BiHoRi04b} and \citet{HoRi05} that this structure forms a strict quantization. And importantly for what follows, the choice of a complex inner product $\alpha$ does not matter at this stage because for any other complex inner product $\alpha'$, the corresponding quantization maps are equivalent in the sense that
\[\lim_{\hbar\to 0}\norm{\mathcal{Q}^{\alpha}_\hbar(A)-\mathcal{Q}^{\alpha'}_\hbar(A)}_\hbar = 0\]
for each $A\in\mathcal{P}$.  For this reason, we will omit any mention of the inner product $\alpha$ and simply denote the quantization map by $\mathcal{Q}_\hbar$, noting that it need not be identified with any of the inner products $\alpha_J$ associated with the complex structures $J$ used to define the number operators of interest.  In this way, the structure of a strict quantization will allow us to take the classical limit independently of the choice of a Fock space representation (which, recall, is determined by a choice of complex structure $J$).

Notice, however, that since the domain of the quantization maps is a C*-algebra of bounded quantities, such structures cannot immediately be used to analyze the classical limits of number operators, which are unbounded.  To analyze the classical limits of number operators, one must extend the quantization map to a larger (partial) algebra including unbounded quantities.  This is possible because the quantization maps $\mathcal{Q}_\hbar$ are positive, and hence continuous, so they can be extended to the completion of $\mathfrak{A}_\hbar$ in the weak topology, which will contain many unbounded operators.

We require one further technical alteration to the setup because the completion of $\mathfrak{A}_\hbar$ will not contain the unbounded field operators defined by 
\begin{align}
\label{eq:fielddef}
\Phi_\hbar(F) := -i\lim_{t\to 0} \frac{W_\hbar(tF)-I}{t}
\end{align}
for $F\in E$ (where the limit is taken in the weak topology).  Instead, one must employ a different algebra.  Define $V_\hbar$ as the subspace of $(\mathfrak{A}_\hbar)^*$ generated by the regular states, and let $\overline{V}_\hbar$ denote the weak* closure of $V_\hbar$ in $(\mathfrak{A}_\hbar)^{***}$.  We will consider the algebras $\mathfrak{A}_\hbar^{**}/N(\overline{V}_\hbar)$, where $N(\overline{V}_\hbar)$ denotes the closed two-sided ideal given by the annihilator of $\overline{V}_\hbar$ in $\mathfrak{A}^{**}_\hbar$.  The quantization maps, which we continue to denote $\mathcal{Q}_\hbar$, can be continuously extended to maps on the completions of these algebras in the weak* topology, which contain unbounded field operators defined by Eq. (\ref{eq:fielddef}) (with the limit in the weak* topology) \citep{BrFeGaLiSo20}.

With this framework, one can define creation, annihilation and number operators by a choice of complex structure $J$ on $E$ according to Eq. (\ref{eq:numberdef}).  \citet{BrFeGaLiSo20} show that the number operators so defined satisfy analogues of the limiting conditions (i) and (ii) stated above for algebraic operations with even many unbounded operators of interest.  Hence, one can understand the classical quantity
\[N_0^J(F) = (a_0^J(F))^*(a_0^J(F)) = \Phi_0(F)^2 + \Phi_0(JF)^2\]
for $F\in E$ as the classical limit of the corresponding number operator (for the complex structure $J$) in a quantum field theory.  For our purposes, the classical limit of the number operator for the two particular choices of $J$ described above are relevant: the Minkowski and Rindler number operators.

The Minkowski number operators are defined using the Minkowski complex structure $J_M$ via Eq. (\ref{eq:numberdef}) as $N^{J_M}_\hbar(F)$ for $F\in E$, which we simply denote $N^M_\hbar(F)$.  These are the standard number operators appearing in the Fock space representation $\pi_M$ for the Minkowski vacuum.  In the framework of the strict quantization defined by Eq. (\ref{eq:quantmap}), the classical limits of the Minkowski number operators take the form $N_0^{M}(F)$, which belongs to the weak* completion of $\mathfrak{A}_0^{**}/N(\overline{V}_0)$.

One can put the classical limit of the Minkowski number operator in a more explicit form.  First, we define the total Minkowski number operator as\footnote{Pointwise convergence on $E\subseteq E'$ (or $E(R)\subseteq E(R)'$) of the infinite sums employed in this section is guaranteed by the cited results of \citet{BrFeGaLiSo20}.}
\begin{align}
    \overline{N}^M_0 := \sum_k N^M_0(F_k),
\end{align}
where $\{F_k\}$ is an $\alpha_{J_M}$-orthonormal basis for $E$.  The classical Weyl algebra $\mathcal{W}(E,0)$ has a canonical representation as continuous almost periodic functions on the dual $E'$ with
\begin{align}
W_0(f,g)(\pi,\varphi) &= \exp\Big(i\int_{\Sigma}\pi f + \varphi g\Big)\\
\Phi_0(f,g)(\pi,\varphi) &= \int_{\Sigma}\pi f + \varphi g \nonumber
\end{align}
for $(f,g),(\pi,\varphi)\in C^\infty(\Sigma) \oplus C^\infty(\Sigma)\subseteq E'$.  \citet{BrFeGaLiSo20} establish the form of the classical limit $\overline{N}_0^M$ of the total Minkowski number operator in this representation.

\begin{thm}
[\citeauthor{BrFeGaLiSo20}]
\label{thm:totnum}
For all $\pi,\varphi\in C_c^\infty(\Sigma)$,
\[\overline{N}_0^M (\pi,\varphi) = \frac{1}{2}\int_{\Sigma} (\mu_M^{-1/2}\pi)^2 + (\mu_M^{1/2}\varphi)^2.\]
\end{thm}

\noindent This shows that the classical limit of the total Minkowski number operator is the integral of a density function depending on the field and its conjugate momentum.

For comparison in the next section, we note that one can use the same methods to analyze the classical limit of the total Hamiltonian.  The total Hamiltonian is defined as
\begin{align}
    H_0^M:=\sum_k N_0^M(g_k,0),
\end{align}
where $\{g_k\}$ is an $L^2(\Sigma,\mathbb{R})$-orthonormal basis for the real vector space $C_c^\infty(\Sigma)$.  In the representation of the classical Weyl algebra as almost periodic functions, we have the following explicit form for the classical limit $H_0^M$ of the total Hamiltonian.

\begin{thm}[\citeauthor{BrFeGaLiSo20}]
For all $\pi,\varphi\in C_c^\infty(\Sigma)$,
\[H_0^M (\pi,\varphi) = \frac{1}{2}\int_{\Sigma} \pi^2 + m^2\varphi^2 + (\nabla\varphi)^2.\]
\end{thm}
\noindent Since the expression on the right hand side is the familiar $(0,0)$ component of the stress-energy tensor for the Klein-Gordon field, this shows that the classical limit of the total Hamiltonian is the classical total energy, which similarly is the integral of an energy density.

We can give a similar analysis of the classical limit of the Rindler number operator. First, we use the Rindler complex structure $J_R$ to define the Rindler number operators via Eq. (\ref{eq:numberdef}) as $N_\hbar^{J_R}(F)$ for $F\in E(R)$, which we simply denote $N_\hbar^R(F)$.   This is the standard number operator appearing in the Fock space representation $\pi_R$ for the Rindler vacuum \citep[see][]{Ka85}.  In the framework of the strict quantization defined by Eq. (\ref{eq:quantmap}), the classical limits of the Rindler number operators take the form $N_0^R(F)$, which belongs to the weak* completion of $\mathfrak{A}_0^{**}/N(\overline{V}_0)$.

As above, one can put the classical limit of the Rindler number operator in a more explicit form.  Define the total Rindler number operator by
\begin{align}
    \overline{N}_0^R := \sum_k N_0^R(F_k),
\end{align}
where $\{F_k\}$ is an $\alpha_{J_R}$-orthonormal basis for $E(R)$.  \citet{BrFeGaLiSo20} establish the form of the total Rindler number operator in the representation of the classical Weyl algebra as almost periodic functions on $E(R)'$.

\begin{thm} [\citeauthor{BrFeGaLiSo20}]
For all $\pi,\varphi\in C_c^\infty(R)$,
\[\overline{N}_0^R(\pi,\varphi) = \frac{1}{2}\int_{R}(\mu_R^{-1/2}e^x\pi)^2 + (\mu_R^{1/2}\varphi)^2.\]
\end{thm}
\noindent This expression is, of course, distinct from that for the Minkowski number operator.  Nevertheless, the classical limit of the Rindler number operator is still the integral of a density function depending on the field and its conjugate momentum.

Further, the same methods can be used to analyze the classical limit of the total Rindler Hamiltonian.  The total Rindler Hamiltonian is defined as
\begin{align}
    H_0^R:=\sum_kN_0^R(g_k,0),
\end{align}
where $\{g_k\}$ is an $L^2(R,\mathbb{R})$-orthonormal basis for $C_c^\infty(R)$.  In the representation of the classical Weyl algebra as almost periodic functions, we have the following explicit form for the classical limit $H_0^R$ of the total Rindler Hamiltonian.

\begin{thm}
[\citeauthor{BrFeGaLiSo20}]
For all $\pi,\varphi\in C_c^\infty(R)$,
\[H_0^R(\pi,\varphi) = \frac{1}{2}\int_Re^{2x}\bigg(\pi^2 + m^2\varphi^2 + \Big(\frac{\partial\varphi}{\partial y}\Big)^2 + \Big(\frac{\partial\varphi}{\partial z}\Big)^2 + e^{-2x}\Big(\frac{\partial\varphi}{\partial x}\Big)^2\bigg).\]
\end{thm}
\noindent Since the expression on the right is the integral of the Rindler energy density associated with the Lorentz boost symmetries of the Rindler wedge \citep[see][]{Ka85}, this shows that the classical limit of the total Rindler Hamiltonian is the classical total Rindler energy.

In summary, we can understand the classical limits of both Minkowski and Rindler number operators and their associated Hamiltonians in a mathematically rigorous framework.  The classical limits of the Hamiltonians correspond to classical energy quantities that can be expressed as the integral of a familiar energy density.  And the classical limits of both number operators correspond to functions of the field and conjugate momentum that can likewise be expressed as the integral of a density function.\footnote{\citet{BrFeGaLiSo20} extend the same analysis to the electromagnetic field, in which case the classical limit of the number operator can also be written as the integral of a number density, which agrees with the proposal for classical photon number appearing in \citet{Se18,Se19}.}    The next section uses these results to aid our understanding of particles in light of the issues outlined in \S\ref{sec:back}.

\section{Emergent Localizable Particles}
\label{sec:part}

We now attempt to provide a positive account of the particle content of Klein-Gordon theory by interpreting the classical number operators and understanding them as approximations to quantum number operators.  To understand the localization properties of these classical number operators, we begin with an analogy to classical localizable energy quantities, which will motivate our approach to analyzing classical number operators.

Notice that there is a standard way to interpret the classical Hamiltonian as giving rise to localized energy quantities.  The classical limit of the Hamiltonian for the Klein-Gordon field is the integral of the standard energy density.  One can use this density to define local energy quantities $H_0^M(\Delta)$ for spatial sets $\Delta\subseteq\Sigma$ in the classical theory by restricting the integral of the energy density to the domain $\Delta$ as follows:
\begin{align}
\label{eq:locen}
H_0^M(\Delta)(\pi,\varphi) := \frac{1}{2} \int_\Delta \pi^2 + m^2\varphi^2 + (\nabla\varphi)^2
\end{align}
for $\pi,\varphi\in C_c^\infty(\Sigma)$.  Eq. (\ref{eq:locen}) provides the standard and familiar way in a classical field theory of describing the amount of localized energy within $\Delta$ associated with a field configuration.

With the definition of classical local energy in mind, we will proceed to provide two interpretations of the classical number of particles in which the total number of particles is the integral of a density function, and hence is localizable in precisely the same sense as energy.  Schematically, we will write this as
\begin{equation}
\label{eq:locschema}
    N_0(\Delta)(\pi,\varphi)\sim \int_{\Delta} n(\pi,\varphi)
\end{equation}
for a density function $n(\pi,\varphi)$, where $\Delta\subseteq\Sigma$ is a spatial set.  We will provide two candidates for the density function, which give rise to different particle interpretations of the classical Klein-Gordon theory.  On the first interpretation, which we call the ``Local Density" interpretation and describe in \S\ref{sec:loc}, we set the total integral of $n$ to be the classical total number operator $\overline{N}_0^M$ and take the density function to be that provided in Thm. \ref{thm:totnum}.  However, this yields an interpretation that is non-standard from the perspective of the quantum field theory by ignoring the role of particle modes, and so we provide a second interpretation based on particle modes.  We describe this interpretation, which we call the ``Uniform Density" interpretation, in \S\ref{sec:uni}, according to which the density function $n$ is itself the classical total number operator $\overline{N}_0^M$, understood as a uniform spatial density.  In this section, we treat only Minkowski particles, but we signal later some directions for analyzing Rindler particles.

\subsection{Local Density}
\label{sec:loc}
In this section, we will define the Local Density interpretation of the classical total number operator, compare it to the Newton-Wigner representation of the number operator in quantum field theory, and then establish that it obeys conditions analogous to those laid out by \citet{HaCl02} for systems of local number operators.

With the definition of classical local energy in mind, now consider the classical total number operator, which Thm. \ref{thm:totnum} shows can also be written as the integral of a density function.  This implies that we can in an analogous way define classical local number quantities $N_0^{LD}(\Delta)$ ($LD$ denotes ``local density'') by restricting the integral of the number density to the domain $\Delta$ as follows.  First, let\footnote{Here, $L^1(\Sigma)$ is the collection of integrable functions on $\Sigma$.  Except where explicitly noted as in \S\ref{sec:quant}, we understand all $L^p$ spaces to consist in \emph{complex-valued} functions.} $n:C^\infty(\Sigma)\oplus C^\infty(\Sigma)\to L^1(\Sigma)$ be such that for all $(\pi,\varphi)\in C^\infty(\Sigma)\oplus C^\infty(\Sigma)$,
\begin{align*}
    \overline{N}^M_0(\pi,\varphi) = \int_\Sigma n(\pi,\varphi).
\end{align*}

Note that $n$ is underdetermined here. We could, for example, take 
\begin{align}
\label{eq:locdens}
    n(\pi,\varphi) := \frac{1}{2}\Big((\mu_M^{-1/2}\pi)^2 + (\mu_M^{1/2}\varphi)^2\Big)
\end{align}
as it is written in Thm. \ref{thm:totnum}.  But we could also take 
\begin{align}
\label{eq:altlocdens}
    n'(\pi,\varphi) := \frac{1}{2}\Big(\pi(\mu_M^{-1}\pi) + \varphi(\mu_M\varphi)\Big)
\end{align}
since $\mu_M$ is self-adjoint and positive.  Which density is appropriate for physics is a substantive matter, but it makes no difference for most of the discussion that follows.  As such, we only assume that we have chosen one such $n$ for our Local Density interpretation; really we have a family of different interpretations for different densities with the same total integral.

With a choice of number density $n$, we fill in the schema of Eq. (\ref{eq:locschema}) by defining
\begin{align}
\label{eq:locnum}
    N_0^{LD}(\Delta)(\pi,\varphi) := \int_\Delta n(\pi,\varphi)
\end{align}
for $\pi,\varphi\in C_c^\infty(\Sigma)$.  This association of local number quantities in the classical field theory with local regions completes what we call the Local Density interpretation.

To better understand this interpretation, let us compare with what \citet[][p. 121-2]{Ha01a} proposes as the natural way in the corresponding quantum theory to construct a local number operator associated with a spatial region $\Delta$ from the number operators $N_\hbar^M(f)$ for test functions $f\in E$.  Halvorson interprets the test functions $f$ in the standard way as one-particle wavefunctions and interprets $N_\hbar^M(f)$ as the number of particles in a state with the wavefunction $f$.  On this interpretation, he suggests that the \emph{standard local number operator} associated with $\Delta$, which we call $N^S_\hbar(\Delta)$ should be the sum\footnote{In order for this sum to converge in the quantum theory ($\hbar>0$), one should actually understand it as an upper bound of quadratic forms in the Minkowski Fock space representation.} of number operators for one-particle wavefunctions with support in $\Delta$, i.e., he defines $N^S_\hbar(\Delta)$ for a spatial set $\Delta$ as
\begin{align}
\label{eq:standardloc}
    N^S_\hbar(\Delta):=\sum_ k N_\hbar^M(f_k),
\end{align}
where $\{f_k\}$ is a basis for the real vector space $C^\infty(\Delta)\oplus C^\infty(\Delta)$.

However, Halvorson shows, contra \citet{Re95b}, that the aforementioned definition does not yield appropriate local number operators.  Rather, for each spatial set $\Delta$, the defined sum yields the total number operator, i.e.,
\[N^S_\hbar(\Delta) = \overline{N}_\hbar^M\]
for any $\Delta$.  This follows immediately from a result of \citet{SeGo65} that the operator $\mu_M$ has the property they call \emph{anti-locality}.

Notice that our natural definition of the local number of particles in the classical field theory does \textit{not} agree with Halvorson's attempted definition of the local number operator in the quantum theory as a sum over the number operators associated with test functions with support in $\Delta$.   Even in the classical case, Halvorson's prescription yields the \emph{total} number operator. More precisely, suppose we define the standard local number operator in the classical theory as
\begin{align}
    N_0^S(\Delta):=\sum_k N^M_0(f_k),
\end{align}
where $\{f_k\}$ is a basis for the real vector space $C^\infty(\Delta)\oplus C^\infty(\Delta)$.  Then the anti-locality of $\mu_M$ established by  \citet{SeGo65} implies the standard local number operator is the total number operator for arbitrary local regions, i.e. for any spatial set $\Delta$,
\begin{align*}
    N_0^S(\Delta) = \overline{N}_0^M
\end{align*}
However, the local number quantities we have defined in Eq. (\ref{eq:locnum}) for the Local Density interpretation are clearly \textit{not} the total classical number operator, i.e.,
\[N_0^{LD}(\Delta)\neq \overline{N}_0^M,\] and so the local number quantities $N_0^{LD}(\Delta)$ we have defined cannot be obtained by summing over a basis of test functions with support in $\Delta$. Nevertheless, the local number quantities defined in Eq. (\ref{eq:locnum}) seem well-motivated in the classical theory.

In fact, the association of local number operators defined by Eq. (\ref{eq:locnum}) instead align with what Halvorson calls the \emph{Newton-Wigner localization scheme}\footnote{For background on the Newton-Wigner localization scheme, see \citet{NeWi49}, \citet{FlBu99}, and \citet{Fl00}.} in the quantum field theory \citep[][p. 123]{Ha01a}, which is obtained by transforming the test function space before summing over test functions with support in a region.  To start, we define a dense isometric embedding $K: E\to L^2(\Sigma)$ (where $E$ is understood with the inner product $\alpha_{J_M}$) by\footnote{See \citet[][p. 65]{Ka85} or \citet[][p. 115, Eq. 10]{Ha01a}.}
\begin{align}
    K(f,g) := \mu_M^{1/2}f + i\mu_M^{-1/2}g
\end{align}
for all $f,g\in C_c^\infty(\Sigma)$.  The Newton-Wigner localization scheme in the quantum theory for $\hbar>0$, by definition, understands $N_\hbar^M(F)$ for $F\in E$ to be local to a region $\Delta\subseteq \Sigma$ just in case $K(F)\in L^2(\Delta)\subseteq L^2(\Sigma)$, i.e., the points outside $\Delta$ for which $K(F)$ is non-zero form a region of measure zero.  Similarly, the Newton-Wigner total number operator in the quantum theory for the region $\Delta$ is then defined as\footnote{Again, in order for this sum to converge in the quantum theory ($\hbar>0$), one should understand it as an upper bound of quadratic forms in the Minkowski Fock space representation.}
\begin{align}
    N_\hbar^{NW}(\Delta) := \sum_k N_\hbar^M(F_k),
\end{align}
where $\{K(F_k)\}$ is an $L^2$-orthonormal basis for $L^2(\Delta)$ \citep[][p. 124]{Ha01a}.

We now define a corresponding Newton-Wigner localization scheme in the classical theory and show that it agrees with our Local Density interpretation.  For this, we must fix that our density is given by Eq. (\ref{eq:locdens}) rather than Eq. (\ref{eq:altlocdens}) or some other alternative.  Using the straightforward analogy provided in the $\hbar\to 0$ limit, we say the Newton-Wigner localization scheme in the classical theory, by definition, understands $N_0^M(F)$ for $F\in E$ to be local to a region $\Delta\subseteq\Sigma$ just in case $K(F)\in L^2(\Delta)\subseteq L^2(\Sigma)$.  And we define the Newton-Wigner total number operator in the classical theory for the region $\Delta$ as\footnote{The pointwise convergence on $E\subseteq E'$ of this sum is guaranteed by the following Prop. \ref{prop:NW}.}
\begin{align}
    N_0^{NW}(\Delta) :=\sum_k N_0^M(F_k),
\end{align}
where $\{K(F_k)\}$ is an $L^2$-orthonormal basis for $L^2(\Delta)$.  Then it follows that our local density interpretation reproduces the Newton-Wigner total number operators in the classical theory.

\begin{prop}
\label{prop:NW}
For any $\Delta\subseteq \Sigma$ and any $\pi,\varphi\in C_c^\infty(\Sigma)$, if $N_0^{LD}(\Delta)$ is defined using the density in Eq. (\ref{eq:locdens}), then
\begin{align*}
    N_0^{LD}(\Delta)(\pi,\varphi) = N_0^{NW}(\Delta)(\pi,\varphi).
\end{align*}
\end{prop}

\begin{proof}
The Pythagorean theorem on $L^2(\Delta)$ implies that,\footnote{Here, we use $\inner{\cdot}{\cdot}_{L^2(X)}$ to denote the inner product in $L^2(X)$.} taking $\{F_k\} = \{f_k,g_k\}$ to be an $L^2$- orthonormal basis for $L^2(\Delta)$,
\begin{align*}
    N_0^{NW}(\Delta)(\pi,\varphi) &= \frac{1}{2}\sum_k\bigg[\big(\int_\Sigma\pi f_k + \varphi g_k\big)^2 + \big(\int_\Sigma (\mu_M^{1/2}\varphi)(\mu_M^{1/2}f_k) - (\mu_M^{-1/2}\pi)(\mu_M^{-1/2}g_k)\big)^2\bigg]\\
    &=\frac{1}{2}\sum_k \abs{\inner{K(\varphi,-\pi)}{K(F_k)}_{L^2(\Sigma)}}^2\\
    &= \frac{1}{2}\sum_k \abs{\inner{K(\varphi,-\pi)}{K(F_k)}_{L^2(\Delta)}}^2\\
    &= \frac{1}{2}\inner{K(\varphi,-\pi)}{K(\varphi,-\pi)}_{L^2(\Delta)}\\
    &=\frac{1}{2}\int_\Delta (\mu_M^{1/2}\varphi)^2 + (\mu_M^{-1/2}\pi)^2\\
    &= N_0^{LD}(\Delta)(\pi,\varphi),
\end{align*}
which is what we set out to show.  \end{proof}
\noindent This establishes that comparison with the construction of local energy quantities suggests an interpretation of number operators in the classical theory---which we have called the Local Density interpretation---which aligns with the Newton-Wigner localization scheme.\footnote{An anonymous referee has asked if one can use the methods in \citet{BrFeGaLiSo20} to analyze the classical limit of the Newton-Wigner \emph{position} operator directly.  If successful, this might provide a way around Malament's ``no-go" theorem in the classical limit.  But the methods for taking the classical limit described here are applicable only to the Weyl algebra, so these methods willow allow one to take the classical limit of the Newton-Wigner position operator only if the Newton-Wigner position operator can itself be related to limits of algebraic combinations of Weyl unitaries.  This is an interesting question, but an answer is beyond the scope of the current paper and so we leave it for future work.}  So despite the negative pronouncements by \citet[][e.g., p. 131-132]{Ha01a}, the classical limit may provide some reason to favor the Newton-Wigner localization scheme, although we take no stance on whether such reasons apply in the corresponding quantum theory.  (We remind the reader at this point that our interpretation of the classical number operators is meant to aid in the understanding of quantum field theory only by showing what a world approximately described by quantum field theory would approximately be like---not by showing what a world described exactly by quantum field theory would be like.  As such, we have not provided any reason to advocate for the Newton-Wigner localization scheme in quantum field theory.)

Moreover, we can also see that the quantities $N_0^{LD}(\Delta)$ are at least appropriate \emph{candidates} for local number operators by demonstrating that the assignment $\Delta\mapsto N_0^{LD}(\Delta)$ satisfies analogs of the necessary conditions (1)-(3) and (5)-(6) of Halvorson and Clifton's no-go theorem.  We first need some preliminaries.  Since we work with initial data on a spacelike hypersurface $\Sigma\cong\mathbb{R}^3$ embedded in $M$, we will use the orthogonal decomposition of each vector $a$ in $M$ into
\[a = t\xi + \eta,\]
where $\xi$ is the unit future-directed timelike vector orthogonal to $\Sigma$, $t\in\mathbb{R}$ is a scalar, and $\eta$ is the spacelike component of $a$ tangent to the hyperplane $\Sigma$.  Recall that the one-parameter unitary family $e^{-i\mu_M t}$ for $t\in\mathbb{R}$ implements the dynamical evolution for the Klein-Gordon field for inertial time translations \citep[][p. 65]{Ka85}.  We extend the assignment $\Delta\subseteq\Sigma\mapsto N_0^{LD}(\Delta)$ to arbitrary spatial sets $\Delta\nsubseteq\Sigma$ (we assume $\Sigma$ belongs to the collection of spacelike hyperplanes foliating $M$) by defining
\begin{align}
    N_0^{LD}(\Delta)(\pi,\varphi):= N_0^{LD}(\Delta - t\xi)(e^{-i\mu_Mt}\pi,e^{-i\mu_Mt}\varphi)
\end{align}
where $t\xi$ is the unique vector beginning on $\Sigma$ and ending on $\Delta$ that is orthogonal to $\Sigma$.

Now we consider Halvorson and Clifton's conditions, beginning with Translation Covariance.  For each vector $a = t\xi + \eta$ define the map $\beta^a: E'\to E'$ acting by
\begin{align}
    \beta^a(\pi,\varphi) := (e^{-i\mu_M t}\eta_*\pi,e^{-i\mu_M t}\eta_*\varphi),
\end{align}
for all $(\pi,\varphi)\in C_c^\infty(\Sigma)\oplus C_c^\infty(\Sigma)$, where $\eta_*$ is the pushforward for the translation by the vector $\eta$.\footnote{One can extend $\beta^a$ to distributional field configurations in $E'$ in an obvious way, but we will only be concerned with the values of number operators on smooth field configurations.}  This allows us to define automorphisms $\beta_a: \mathcal{W}(E,0)\to \mathcal{W}(E,0)$ by linearly extending
\begin{align}
    \beta_a(W_0(f))(\pi,\varphi) := W_0(f)(\beta^a(\pi,\varphi)),
\end{align}
which likewise extends to the number operators as
\begin{align}
    \beta_a(N_0^M(f))(\pi,\varphi):= N_0^M(f)(\beta^a(\pi,\varphi))
\end{align}
for all $f\in E$ and
\begin{align}
    \beta_a(n)(\pi,\varphi) := n(\beta^a(\pi,\varphi)).
\end{align}
It follows from the linearity of $\beta_a$, the translation invariance of $\mu_M$, and the translation invariance of the measure on $\Sigma$ defining the integral that
\begin{align}
    \beta_a(N_0^{LD}(\Delta)) = N_0^{LD}(\Delta+a),
\end{align}
which is an analogue of Translation Covariance in the classical theory.

The classical theory furthermore satisfies the Energy Condition because $\mu_M$ is the self-adjoint generator of the group of inertial time translations and its spectrum is $[m,\infty)$, which is bounded from below.  One might worry that the positivity of the operator $\mu_M$ does not provide an appropriate notion of energy positivity in the classical theory.  But the classical energy density of the Klein-Gordon theory, which is the generator of time translations in the distinct sense of being the associated conserved quantity via Noether's theorem, is also positive in the analogous sense of satisfying the weak energy condition.

Microcausality is trivially satisfied because the algebra of observables for the classical theory is commutative.  One might wish to investigate further the correlations between the expectation values of local number operators associated with spacelike separated regions.  It follows from the properties of $\mu_M$ that $N_0^{LD}(\Delta)$ can, in some sense, depend on the values of $\varphi$ and $\pi$ outside of the region $\Delta$.  This may provide difficulties for understanding the local number operators as quantities measurable by probing the field within the region $\Delta$.  An analogous feature rears its head in the discussion \citet[][p. 128]{Ha01a} gives of the Newton-Wigner local number operators in the quantum theory, which he shows do not satisfy Microcausality and thereby admit ``act-outcome" correlations at spacelike separation.  We leave such considerations for future work, but it is unclear whether there is reason for thinking that this non-local dependence of the classical number operators on the fields violates the constraints of relativity theory.

Linearity of the integral implies that Additivity is satisfied.  Similarly, the countable additivity of the integral implies that for any countable disjoint covering $\{\Delta_n\}$ of $\Sigma$,
\begin{align}
    \overline{N}^{LD}_0 = \sum_n N_0^{LD}(\Delta_n),
\end{align}
which means that the sum\footnote{Again, this sum should be understood in the sense of pointwise convergence on $E\subseteq E'$.} of local number operators converges to the total number operator.  To establish Number Conservation, note that when $\pi,\varphi\in C_c^\infty(\Sigma)$ are smooth field configurations, we can treat $(\pi,\varphi)$ as an element of $E$ for which
\[\overline{N}_0^{LD}(\pi,\varphi) = \overline{N}_0^M(\pi,\varphi) = \frac{1}{2}\alpha_{J_M}((\varphi,-\pi),(\varphi,-\pi)).\]
Further, for any vector $a$, the operator $\beta^a$ on $E$ is a unitary operator (with respect to $\alpha_{J_M}$).  It follows from these two facts that $\beta_a(\overline{N}_0^{LD}) = \overline{N}_0^{LD}$ for any timelike vector $a$.  This shows that the total number operator is conserved under time translations, so an analogue of the Number Conservation is satisfied.  Thus, the classical local number quantities satisfy all of Halvorson and Clifton's conditions, which they claim are at least necessary conditions for an assignment of local number quantities.

The Local Density interpretation allows one to associate a quantity $N_0^{LD}(\Delta)$ with each region $\Delta$, which one can interpret physically as the number of particles in the region $\Delta$ according to the classical theory.  Saying that particles are localizable in this sense is only to say that these quantities satisfy the weak necessary conditions for associating physical quantities with regions of a relativistic spacetime.  This does not, of course, answer any questions about the ontological status of particles in any fundamental quantum field theory.  Instead, it answers the question outlined in the introduction about what the world is approximately like if quantum field theory is approximately accurate: the world would be approximately described by a classical field theory, admitting an interpretation of certain quantities associated with the fields as the number of particles in a region.

Thus, we believe the Local Density interpretation of the number operators in the classical limit provides a sense in which number operators in the quantum field theory are \emph{approximately localizable}.  The sense in which the number operators in the quantum field theory are approximately localizable is captured by the fact that they approximate classical number operators by satisfying the limiting conditions of strict quantization.  Further, the classical number operators on the Local Density interpretation are localizable because (i) they satisfy Halvorson and Clifton's conditions, and (ii) their localization scheme matches standard ways of understanding the localizability of global quantities like energy that can be expressed as the integral of a density function.  Thus, the Local Density interpretation provides one sense in which localizable number operators emerge in the classical limit of quantum field theory.

\subsection{Uniform Density}
\label{sec:uni}

In this section, we will define an alternative interpretation of particles in the classical Klein-Gordon theory---which we call the Uniform Density interpretation of the classical total number operator.  We will motivate this interpretation by showing the classical total number operator can be obtained as the sum over modes of the particle contents naturally associated with the Fourier modes of a Klein-Gordon field.

One might be dissatisfied with the Local Density interpretation of the previous section \emph{precisely because} its fails to match the natural prescription Halvorson describes as the standard way of defining local number operators in the quantum theory, which we reviewed in the previous section.  The standard prescription Halvorson describes for associating number operators with local regions is, after all, based on the standard interpretation of particle modes in the quantum theory.  In this section, we show that one can stay much closer to the prescriptions suggested in the quantum theory when defining classical number quantities by analyzing classical particle modes.  Doing so yields a distinct interpretation, which we call the Uniform Density interpretation.  We give the interpretation this name because we will show that it implies that the total classical number operator $\overline{N}_0^M$ should be understood as a uniform spatial density.  Using $N_0^{UD}(\Delta)$ to denote the number of particles in a spatial region $\Delta$, the Uniform Density interpretation fills in Eq. (\ref{eq:locschema}) by prescribing the association
\begin{align}
    N_0^{UD}(\Delta)(\pi,\varphi) := \int_{\Delta}\overline{N}_0^M(\pi,\varphi).
\end{align}
This association satisfies Halvorson and Clifton's conditions on systems of local number operators in just the same way as the Local Density interpretation, so we will not show this explicitly.  Instead, we focus on motivating the Uniform Density interpretation.

Our motivation comes from the standard way of thinking in quantum field theory according to which different number operators $N^M_\hbar(f)$ correspond to the number of particles in different ``modes" of the field, where modes are understood as components of the Fourier decomposition (in turn corresponding to the inertial timelike symmetries we used to define $\mu_M$ to begin with).  Of course, we can apply this approach to the classical field theory using standard methods of Fourier analysis.  We will proceed by noticing that each Fourier mode of a classical field can be reinterpreted as a relativistic fluid associated with a constant particle number density.  Then we will show that summing the corresponding constant number densities over all modes reproduces the classical total number quantity $\overline{N}_0^M$, thus motivating our understanding it as a constant density.

First, some preliminaries.  To define the Fourier transform, we fix an arbitrary origin to Minkowski spacetime $o\in M$ and understand the position vector of any other point $p\in M$ to be the unique vector $x$ such that $p = o + x$.  (We will show later that the choice of origin does not make a difference to what follows.)  For any scalar field $\varphi\in \mathcal{S}(M)$,\footnote{Here, $\mathcal{S}(M)$ denotes the Schwartz space of rapidly decreasing functions on $M$.} define the Fourier transform, denoted $\tilde{\varphi} \in \mathcal{S}(M)$ by
\begin{align}
    \tilde{\varphi}(k) := \frac{1}{(2\pi)^2}\int_{M}\varphi(x)e^{-ig(k,x)}\ dx
\end{align}
for all vectors $k$ in $M$, where $g$ denotes the Minkowski metric.\footnote{We work with signature $(+,-,-,-)$ for the Minkowski metric $g$.}  The Fourier inversion theorem implies that any solution $\varphi$ to the Klein-Gordon equation (for fixed mass $m$) can be written in the form
\begin{align}
    \varphi(x) = \frac{1}{(2\pi)^2}\int_{S} \tilde{\varphi}(k)e^{ig(k,x)}\ dk,
\end{align}
where the domain of the integral is the \emph{mass shell} $S:=\{k\ |\ g(k,k) = m^2\}$.

Moreover, since $\varphi$ is real, we have $\tilde{\varphi}(-k) = \overline{\tilde{\varphi}(k)}$, which implies that $\varphi$ takes the form
\[\varphi(x) = \frac{1}{(2\pi)^2}\int_{S\cap I^+} 2 Re(\tilde{\varphi}(k))\cos(g(k,x)) + 2Im(\tilde{\varphi}(k))\sin(g(k,x))\ dk,\]
where $I^+$ is the collection of future-directed timelike vectors.  We call the integrand the \emph{$k$-mode of $\varphi$}, explicitly given by
\begin{align}
 \overset{k}{\varphi}(x) := 2 Re(\tilde{\varphi}(k))\cos(g(k,x)) + 2Im(\tilde{\varphi}(k))\sin(g(k,x)).
\end{align}
We define the $k$-modes of $\varphi$ so that
\begin{align*}
\varphi(x) = \frac{1}{(2\pi)^2}\int_{S\cap I^+} \overset{k}{\varphi}(x)\ dk.
\end{align*}
This provides the familiar sense in which $\varphi$ can be thought of as a sum over modes, where each mode is itself a real-valued solution to the Klein-Gordon equation associated with some $k\in S\cap I^+$.

We now analyze each $k$-mode $\overset{k}{\varphi}$ as a solution to the Klein-Gordon equation in its own right.  It is well known that Klein-Gordon fields can be canonically associated with relativistic perfect fluids \citep[See, e.g.,][]{Ma88}.  We understand the fluid corresponding to a Klein-Gordon field $\varphi$ to be described by three quantities: the energy density $\rho$, pressure density $p$, and velocity field $u$, given by
\begin{align}
    \rho &:= \frac{1}{2}(g(\nabla\varphi,\nabla\varphi) + m^2\varphi^2)\\
    p &:= \frac{1}{2}(g(\nabla\varphi,\nabla\varphi) - m^2\varphi^2)\\
    u &:= \frac{\nabla\varphi}{g(\nabla\varphi,\nabla\varphi)^{1/2}},
\end{align}
where here, and in the remainder of this section, $\nabla$ denotes the unique (Levi-Civita) covariant derivative operator on $M$ compatible with $g$ \citep[See][p. 77, Prop. 1.9.2]{Ma12}.  With these definitions, the stress-energy tensor for the Klein-Gordon field takes the same form as the stress-energy tensor for a perfect fluid.  We denote the corresponding energy density, pressure density, and velocity field for the $k$-mode $\overset{k}{\varphi}$ by $\overset{k}{\rho}$, $\overset{k}{p}$, and $\overset{k}{u}$.

Straightforward formal calculation yields
\begin{align*}
    \overset{k}{\rho} &= 2m^2\overline{\tilde{\varphi}(k)}\tilde{\varphi}(k)\\
    \overset{k}{u} &= \frac{k}{m}.
\end{align*}
There is a technical issue with this result for $\overset{k}{\rho}$ that we return to below.  But for the moment notice that this establishes (at least formally) that the fluid corresponding to the $k$-mode has energy density and velocity that are constant across spacetime.

We can use the quantities associated with the interpretation of a scalar field as a fluid to define a particle density.  First, we understand each $k$-mode as a fluid composed of particles of mass $m$.  A co-moving observer (moving with velocity $\overset{k}{u}$) will thus naturally assign the $k$-mode a particle density $\overset{k}{n}_u$ given by
\begin{align}
\overset{k}{n}_u := \frac{\overset{k}{\rho}}{m}.
\end{align}
(Notice that we work in natural units where the speed of light is $c=1$, which simplifies the correspondence between mass and energy.)  But it is well known that the co-moving particle density $\overset{k}{n}_u$ is not relativistically invariant, in the sense that different observers may assign different particle densities.  Following \citet[][p. 49]{We72}, we define, for each $k$-mode, the invariant particle current density as the vector field
\begin{align}
    \overset{k}{N} := \overset{k}{n}_u \overset{k}{u}.
\end{align}
Indeed, $\overset{k}{N}$ is a conserved quantity in the sense that 
\begin{align*}
    \text{div\ } \overset{k}{N} &= g\big(\frac{\overset{k}{{u}}}{m},\nabla\overset{k}{\rho}\big) + \frac{\overset{k}{\rho}}{m} \text{div\ } \overset{k}{u} = -\frac{\overset{k}{p}}{m}\text{div\ }\overset{k}{u} = 0,
 \end{align*}
where the second to last equality follows from the fact that the stress-energy tensor of the Klein-Gordon field is conserved \citep[][p. 150, Eq. 2.5.5]{Ma12}, and the last equality follows from the fact that the velocity field $\overset{k}{u}$ is constant.

At this point, it is important to note that the particle content of each $k$-mode, encoded in the current density $\overset{k}{N}$, is independent of the choice of origin $o\in M$.  For, if we had fixed a different origin $o\mapsto o' = o + v\in M$, we would only change the Fourier transform by a phase factor
\begin{align*}
    \tilde{\varphi}(k)\mapsto e^{-ig(k,v)}\tilde{\varphi}(k),
\end{align*}
leaving the absolute modulus $\overline{\tilde{\varphi}(k)}\tilde{\varphi}(k)$ invariant.  Thus, a change of origin does not affect the energy density $\overset{k}{\rho}$ or velocity $\overset{k}{u}$ assigned to each $k$-mode, and in turn does not affect the co-moving particle density $\overset{k}{n_u}$ or current density $\overset{k}{N}$.  Hence, although we fixed a choice of origin to define the Fourier transform, the conclusions we draw about particle content are independent of that choice.\footnote{On the other hand, notice that the analysis of this section depends entirely on the use of plane waves as Fourier modes, which are associated with the inertial timelike symmetries of Minkowski spacetime.  This convention \emph{cannot} be changed arbitrarily like the choice of origin.}

Now the particle current for the $k$-mode allows us to define the particle density for the $k$-mode associated with an observer moving with an arbitrary velocity $\xi\in I^+$ as
\begin{align}
    \overset{k}{n_\xi}:= g(\overset{k}{N},\xi).
\end{align}
Since the full field $\varphi$ is itself an integral of the $k$-modes, where $k$ ranges over the future-directed mass shell $S\cap I^+$, we want to associate with an observer moving with velocity $\xi\in I^+$ a total particle density given by the integral of $n_\xi$ over all possible values of $k$.  This will, however, require one technical change.

Our basic strategy is to use $\xi$ to decompose each vector $k$ into a frequency and wave-vector component to put $\overset{k}{n}_\xi$ is a form more amenable to calculations.  To that end, define the frequency $\overset{k}{\omega}$ and wave-number $\mathbf{k}$ relative to the observer with velocity $\xi\in I^+$ by
\begin{align}
\label{eq:obs}
\overset{k}{\omega} := g(\xi,k) && \mathbf{k} = k - g(\xi,k)\xi.
\end{align}
We also define
\begin{align}
   (\mathbf{k}\cdot\mathbf{k}) := -g(\mathbf{k},\mathbf{k}).
\end{align}
so that
\begin{align}
    \overset{k}{\omega} = (m^2 + (\mathbf{k}\cdot\mathbf{k}))^{1/2}.
\end{align}

Similarly, we understand $\Sigma$ to be the spatial slice associated with the observer moving in the direction $\xi$, i.e., we let $\Sigma$ be the collection of spacelike vectors orthogonal to $\xi$.\footnote{Or equivalently, we can let $\Sigma$ consist in points $p\in M$ with $p = o + \mathbf{x}$ for some spacelike vector $\mathbf{x}$ orthogonal to $\xi$.  The relevant quantities are all invariant under time translations, so it does not matter which surface we choose from a spacelike foliation.}  Similarly, given any $\varphi\in C_c^\infty(M)$, we associate with the observer moving in the direction $\xi$ the standard field momentum $\pi = g(\xi,\nabla\varphi)$, using the time derivative with respect to the proper time of the observer.

Then calculating from our definitions with the Fourier transform yields
\begin{align*}
    \overset{k}{n}_\xi &= 2(m^2 + (\mathbf{k}\cdot\mathbf{k}))^{1/2} \cdot \overline{\tilde{\varphi}(k)}\tilde{\varphi}(k)\\
    \tilde{\pi}(k) &= i(m^2 + (\mathbf{k}\cdot\mathbf{k}))^{1/2}\tilde{\varphi}(k).
\end{align*}

It follows that we can write the particle density for the $k$-mode as
\begin{align}
\label{eq:unidens}
    \overset{k}{n}_\xi = (m^2 + (\mathbf{k}\cdot\mathbf{k}))^{1/2}\cdot\overline{\tilde{\varphi}(k)}\tilde{\varphi}(k) + (m^2 + (\mathbf{k}\cdot\mathbf{k}))^{-1/2}\cdot \overline{\tilde{\pi}(k)}\tilde{\pi}(k)
\end{align}

Now the technical issue mentioned earlier is that the products $\overline{\tilde{\varphi}(k)}\tilde{\varphi}(k)$ and $\overline{\tilde{\pi}(k)}\tilde{\pi}(k)$ are not generally well-defined---even as distributions---because they may involve products of delta functions.  To see this, first define the spatial Fourier transform for initial conditions $\varphi\in\mathcal{S}(\Sigma)$ on the spacelike hyperplane $\Sigma$ as
\begin{align}
    \mathcal{F}(\varphi)(\mathbf{k}) := \frac{1}{(2\pi)^2}\int_\Sigma \varphi(\mathbf{x})e^{i\mathbf{k}\cdot\mathbf{x}}\ d\mathbf{x}.
\end{align}
In other words, $\mathcal{F}$ is the ordinary Fourier transform in $\mathbb{R}^3$.  It follows that the spatial Fourier transform is related to the full Fourier transform by
\begin{align*}
    \tilde{\varphi}(k) = \mathcal{F}(\varphi)(\mathbf{k})\cdot \delta(m^2 + \mathbf{k}^2),
\end{align*}
where we use the same convention in Eq. (\ref{eq:obs}) for associating each $k$ with a wave-number $\mathbf{k}$.  So if we have initial conditions $\varphi\in \mathcal{S}(\Sigma)$ (which are required for our definition of the number operator), then since  $\overline{\tilde{\varphi}(k)}\tilde{\varphi}(k)$ and $\overline{\tilde{\pi}(k)}\tilde{\pi}(k)$ contain products of delta functions, then, they are not well-defined.

This provides some motivation for working instead with a formally analogous, but well-defined expression for the number density that is obtained by using only the spatial Fourier transform.  We define, in analogy with Eq. (\ref{eq:unidens}), a ``spatial" number density
\begin{align*}
    \overset{\mathbf{k}}{n}_\xi := (m^2 + (\mathbf{k}\cdot\mathbf{k}))^{1/2}\cdot \overline{\mathcal{F}(\varphi)(\mathbf{k})}\mathcal{F}(\varphi)(\mathbf{k}) + (m^2 + \mathbf{k}\cdot\mathbf{k}))^{-1/2}\overline{\mathcal{F}(\pi)(\mathbf{k})}\mathcal{F}(\pi)(\mathbf{k}).
\end{align*}
This definition in turn allows us to integrate over all $k$-modes by integrating over all possible values of $\mathbf{k}$, understood as associated with the frequency $(m^2 + (\mathbf{k}\cdot\mathbf{k}))^{1/2}$.

As such, we define
\begin{align}
    n_\xi := \int_{S\cap I^+} \overset{\mathbf{k}}{n}_\xi\ d\mathbf{k}.
\end{align}
Since the particle density $\overset{k}{n}_\xi$ in the direction $\xi$ for each $k$-mode is a constant scalar field, so is the total particle density in the direction $\xi$.  In this sense, our interpretation of the $k$-modes of $\varphi$ as perfect fluids yields a local, but \emph{uniform} particle density in spacetime.

Finally, we now establish that, on this definition, the total particle density for an observer with velocity $\xi\in I^+$ agrees with the total classical number operator $\overline{N}_0^M$ at every point.  We have the following result.

\begin{prop}
For any $\varphi\in C_c^\infty(M)$, with the above definitions of $\Sigma$ and $\pi$,
\begin{align*}
    n_\xi(\mathbf{x}) =  \overline{N}_0^M(\pi_{|\Sigma},\varphi_{|\Sigma})
\end{align*}
for all $\mathbf{x}\in \Sigma$.
\end{prop}

\begin{proof}
Recall that \citep[See, e.g.,][p. 50]{ReSi75}
\begin{align*}
(\mathcal{F}(\mu_M\varphi))(\mathbf{k}) = (m^2 + (\mathbf{k}\cdot\mathbf{k}))^{1/2}\cdot (\mathcal{F}\varphi)(\mathbf{k}).
\end{align*}
The Plancherel theorem \citep[][p. 10, Thm. IX.6]{ReSi75} then immediately implies our result:
\begin{align*}
 \overline{N}_0(\pi_{|\Sigma},\varphi_{|\Sigma}) &= \frac{1}{2}\int_{\Sigma}\varphi(\mathbf{x})(\mu_M\varphi)(\mathbf{x}) + \pi(\mathbf{x})(\mu_M^{-1}\pi)(\mathbf{x})\ d\mathbf{x}\\
  &=\frac{1}{2}\int_{S}\overline{\tilde{\varphi}(k)}(m^2 + (\mathbf{k}\cdot\mathbf{k}))^{1/2}\tilde{\varphi}(k) + \overline{\tilde{\pi}(k)}(m^2 + (\mathbf{k}\cdot\mathbf{k}))^{-1/2}\tilde{\pi}(k)\ dk\\
 &= \frac{1}{2}\int_{S}\overset{\mathbf{k}}{n}_\xi\ dk = \int_{S\cap I^+} \overset{\mathbf{k}}{n}_\xi\ dk,
\end{align*}
which is what we set out to show.
\end{proof}

Now we have shown that the classical total number operator $\overline{N}_0^M$ can be thought of as a uniform particle density, according to what we call the Uniform Density interpretation, which assigns to the region $\Delta$ the local particle number $N_0^{UD}(\Delta)$.  On this interpretation, every Klein-Gordon field is understood as a linear combination of perfect fluids, each associated with a constant velocity and energy density, implying a constant particle number and current density.  The classical total number operator $\overline{N}_0^M$ is what we would naturally think of as the total particle number density of the entire fluid as measured by an observer, which we obtain by integrating the particle number density, as seen by that observer, for each of the component fluids (modes).  Since the particle number densities for each of the component fluids are local in the standard sense, we understand the total particle number density $\overline{N}_0^M$ for the entire fluid to be local in the same sense.  Moreover, since $\overline{N}_0^M$ is an approximation to the total number operator in the quantum field theory, this provides a sense in which the quantum field theory allows for approximately localizable particles.

\section{Conclusion}
\label{sec:con}

In this paper, we have argued that, while results of \citet{Ma96} and \citet{HaCl02} provide obstacles to interpreting relativistic quantum field theories in terms of localizable particles, one can gain some understanding of particles in quantum field theories through the approximation of the classical limit.  We reviewed recent results of \citet{BrFeGaLiSo20} establishing the form of the classical limit of number operators in the quantum theory of the free Klein-Gordon field.  Our central contribution was to show that the classical number operators so obtained can be understood as localizable in at least two distinct senses.  First, on our Local Density interpretation, we showed from the fact that the classical total number operator can be written as a spatial integral of a density function that it yields a natural assignment of particle contents to local regions, which satisfies analogs of the conditions that Halvorson and Clifton propose for localization schemes, and agrees with the Newton-Wigner localization scheme when extended to the classical theory.  We noted, however, that there is still a sense in which on this interpretation the value of the number density in any given region depends on the field values outside of the region.  Second, on our Uniform Density interpretation, we showed by decomposing a Klein-Gordon field into its Fourier modes and understanding each mode as a perfect fluid with constant velocity and particle density, that the classical total number operator can also be understood as itself a uniform density obtained by summing the local particle densities associated with all modes.  These two interpretations provide distinct routes to understanding the particle content of the quantum Klein-Gordon theory as approximately local.

We take no stance on which of the two interpretations---Local Density or Uniform Density---is preferable.  We note here though that the density $n$ associated with the Local Density interpretation is not, in general, locally conserved under timelike translations while the density $\overline{N}_0^M$ associated with the Uniform Density interpretation is.  This perhaps speaks in favor of the Uniform Density interpretation for describing particles that persist over time.  On the other hand, this may speak against the Uniform Density interpretation because it disagrees with the description from the full quantum field theory of particles that can be created or annihilated.
We encourage further consideration of these particle interpretations of classical field theory to aid in the understanding of quantum field theory.

We mentioned in \S\ref{sec:loc} that on the Local Density interpretation, the value of the classical number operator in a region in some sense depend on the values of the fields outside that region.  As an aside, we now note that one can use the technical tools previously outlined to display this non-local dependence explicitly---calculating with the Fourier transform yields for the total number operator:
\begin{align}
    \overline{N}_0^M(\pi,\varphi) &= \frac{1}{2}\int_\Sigma \frac{\abs{\mathcal{F}(\pi)}(\mathbf{k})^2+(m^2 + (\mathbf{k}\cdot\mathbf{k}))\cdot \abs{\mathcal{F}(\varphi)(\mathbf{k})}^2}{(m^2 + (\mathbf{k}\cdot\mathbf{k}))^{1/2}} d\mathbf{k}\\
    &= \frac{1}{2(2\pi)^4}\int_\Sigma\int_\Sigma\int_\Sigma\Big(\pi(\mathbf{x})\pi(\mathbf{y}) + (\mu_M\varphi)(\mathbf{x})(\mu_M\varphi)(\mathbf{y})\Big) \frac{1}{(m^2 + (\mathbf{k}\cdot\mathbf{k}))^{1/2}} e^{-i\mathbf{k}\cdot(\mathbf{x}-\mathbf{y})} d\mathbf{k}\ d\mathbf{x}\ d\mathbf{y}\\
    &=\frac{1}{2\pi} \int_\Sigma \int_\Sigma \Big(\pi(\mathbf{x})\pi(\mathbf{y})\  + (\mu_M\varphi)(\mathbf{x})(\mu_M\varphi)(\mathbf{y})\Big)D(\mathbf{x}-\mathbf{y})\ d\mathbf{x}\ d\mathbf{y},
\end{align}
where we follow \citet{PeSc95} in using $D(\mathbf{x}-\mathbf{y})$ to denote the quantum field theoretic correlation function (i.e., the expectation value in the vacuum state of $\Phi(\mathbf{x})\Phi(\mathbf{y})$, where $\Phi$ is the operator valued distribution corresponding to the quantum field $\Phi$.  The first line follows from the Plancherel theorem, the second by subsituting the definition of the Fourier transform, and the third by substituting the definition of the correlation function.  Although this expression displays a dependence of the number operator on field values at spacelike separation, this dependence is somewhat attenuated.  A length scale for the dependence is set by the mass $m$ as follows: \citet[][p. 27]{PeSc95} show the asymptotic behavior of $D(\mathbf{x}-\mathbf{y})$ as $\abs{\mathbf{x}-\mathbf{y}}\to\infty$ is
\begin{align*}
    D(\mathbf{x}-\mathbf{y})\sim e^{-im\abs{\mathbf{x}-\mathbf{y}}},
\end{align*}
which entails that the total number operator is $1/m$-local in the sense of \citet{Wa01}.  This provides perhaps another sense in which the classical number operator is approximately local, but also perhaps a sense in which it is still not as local as one might have hoped.

In interpreting classical number operators, we focused solely on the classical Minkowski number operator and ignored the classical Rindler number operator, which the analysis of \citet{BrFeGaLiSo20} also covers.  The Local Density interpretation can be immediately extended to the Rindler number operators because Browning et al. establish that the classical total Rindler number operator is the spatial integral of a density function.  So the Local Density interpretation is general enough to apply to inequivalent number operators.

However, extending the Uniform Density interpretation to Rindler number operators requires further work.  It is well known that the Rindler number operators are associated with the Lorentz boost symmetries of the Rindler wedge, in the sense that the Rindler complex structure is the unique one commuting with the Lorentz boost symmetries, which are time translations in Rindler coordinates.  This implies in turn that the Fock space associated with the Rindler representation is built out of a different set of modes, again associated with these distinct timelike symmetries \citep{LePf81}.  In turn, these Rindler modes give rise to an alternative decomposition of a classical Klein-Gordon field---that is, an alternative to the decomposition into Fourier modes provided by the Fourier transform.  The Rindler modes can be constructed from Macdonald functions \citep[i.e., modified Bessel functions of the second kind; see][p. 911]{GrRy07}, and the decomposition into Rindler modes is provided by the associated Kontorovich-Lebedev transform \citep[][Ch. 2]{Ya96}.  If one uses the Kontorovich-Lebedev transform to decompose a classical Klein-Gordon field into Rindler modes, can one also recover a Uniform Density interpretation of classical Rindler particle content?  We save this question for future work, although we note that further investigations of these classical particle interpretations may help expose their relative merits, and may aid in the understanding of inequivalent particle concepts in quantum field theory.

Our results, and the further work suggested by them, demonstrate the usefulness of the classical limit as a tool for interpreting quantum field theory.  We hope this provides just a start to a better understanding of the particle content of quantum field theories.

\begin{acknowledgements}
We are especially indebted to Feiyang Liu, who worked closely with us on the mathematical results of this paper.  Thanks also to two anonymous reviewers, Adam Caulton, Kade Cicchella, Jeremy Steeger, James Weatherall, and the audience of the conference ``Foundations of Quantum Field Theory" (Rotman Institute of Philosophy, 2019) for helpful comments and discussion.
BHF acknowledges support during the completion of this work from the Royalty Research Fund at the University of Washington as well as the National Science Foundation under Grant No. 1846560.
\end{acknowledgements}

\bibliographystyle{apalike}
\setcitestyle{authoryear}
\bibliography{bibliography.bib}

\begin{thebibliography}{}

\bibitem[Arageorgis et~al., 2003]{ArEaRu03}
Arageorgis, A., Earman, J., and Ruetsche, L. (2003).
\newblock {Fulling Non-uniqueness and the Unruh Effect: A Primer on Some
  Aspects of Quantum Field Theory}.
\newblock {\em Philosophy of Science}, 70(1):164--202.

\bibitem[Arageorgis and Stergiou, 2013]{ArSt13}
Arageorgis, A. and Stergiou, C. (2013).
\newblock {On Particle Phenomenology Without Particle Ontology: How Much Local
  Is Almost Local?}
\newblock {\em Foundations of Physics}, 43:969--977.

\bibitem[Baker, 2016]{Ba16}
Baker, D. (2016).
\newblock {The Philosophy of Quantum Field Theory}.
\newblock In {\em Oxford Handbooks Online, \emph{\texttt{DOI:
  10.1093/oxfordhb/9780199935314.013.33}}}. Oxford University Press.

\bibitem[Bell, 1987]{Be87}
Bell, J.~S. (1987).
\newblock {\em Speakable and Unspeakable in Quantum Mechanics: Collected Papers
  on Quantum Philosophy}.
\newblock Cambridge University Press, 2 edition.

\bibitem[Bigaj, 2018]{Bi18}
Bigaj, T. (2018).
\newblock Are field quanta real objects? some remarks on the ontology of
  quantum field theory.
\newblock {\em Studies in the History and Philosophy of Modern Physics},
  62:145--157.

\bibitem[Binz et~al., 2004]{BiHoRi04b}
Binz, E., Honegger, R., and Rieckers, A. (2004).
\newblock {Field-theoretic Weyl Quantization as a Strict and Continuous
  Deformation Quantization}.
\newblock {\em Annales de l'Institut Henri Poincar\'{e}}, 5:327--346.

\bibitem[Bratteli and Robinson, 1987]{BrRo87}
Bratteli, O. and Robinson, D. (1987).
\newblock {\em {Operator Algebras and Quantum Statistical Mechanics}},
  volume~1.
\newblock Springer, New York.

\bibitem[Bratteli and Robinson, 1996]{BrRo96}
Bratteli, O. and Robinson, D. (1996).
\newblock {\em {Operator Algebras and Quantum Statistical Mechanics}},
  volume~2.
\newblock Springer, New York.

\bibitem[Browning et~al., 2020]{BrFeGaLiSo20}
Browning, T., Feintzeig, B., Gates, R., Librande, J., and Soiffer, R. (2020).
\newblock {Classical Limits of Unbounded Quantities by Strict Quantization}.
\newblock {\em Unpublished}.

\bibitem[Buchholz, 1995]{Bu95}
Buchholz, D. (1995).
\newblock {On the Manifestations of Particles}.
\newblock {\em Unpublished.}, \texttt{arXiv:hep-th/9511023v1}.

\bibitem[Buchholz and Fredenhagen, 1982]{BuFr82}
Buchholz, D. and Fredenhagen, K. (1982).
\newblock Locality and the structure of particle states.
\newblock {\em Communications in Mathematical Physics}, 84:1--54.

\bibitem[Buchholz et~al., 1991]{BuPoSt91}
Buchholz, D., Porrman, M., and Stein, U. (1991).
\newblock {Dirac versus Wigner: Towards a universal particle concept in local
  quantum field theory}.
\newblock {\em Physics Letters B}, 267(3):377--381.

\bibitem[Buchholz and Verch, 1995]{BuVe95}
Buchholz, D. and Verch, R. (1995).
\newblock Scaling algebras and renormalization group in algebraic quantum field
  theory.
\newblock {\em Reviews in Mathematical Physics}, 7(8):1195.

\bibitem[Buchholz and Verch, 1998]{BuVe98}
Buchholz, D. and Verch, R. (1998).
\newblock Scaling algebras and renormalization group in algebraic quantum field
  theory. ii. instructive examples.
\newblock {\em Reviews in Mathematical Physics}, 10(6):775--800.

\bibitem[Clifton and Halvorson, 2001]{ClHa01}
Clifton, R. and Halvorson, H. (2001).
\newblock {Are Rindler Quanta Real?: Inequivalent Particle Concepts in Quantum
  Field Theory}.
\newblock {\em British Journal for the Philosophy of Science}, 52:417--470.

\bibitem[Earman, 2011]{Ea11}
Earman, J. (2011).
\newblock The unruh effect for philosophers.
\newblock {\em Studies in the History and Philosophy of Modern Physics},
  42:81--97.

\bibitem[Feintzeig, 2019]{Fe19}
Feintzeig, B. (2019).
\newblock {The classical limit as an approximation}.
\newblock {\em Philosophy of Science}, forthcoming,
  \texttt{http://philsci-archive.pitt.edu/16359/}.

\bibitem[Fleming, 2000]{Fl00}
Fleming, G. (2000).
\newblock {Reeh-Schlieder Meets Newton-Wigner}.
\newblock {\em Philosophy of Science (Proceedings of the 1998 Biennial Meetings
  of the PSA)}, 67(3):S495--S515.

\bibitem[Fleming and Butterfield, 1999]{FlBu99}
Fleming, G. and Butterfield, G. (1999).
\newblock Strange positions.
\newblock In Butterfield, J. and Pagonis, C., editors, {\em From Physics to
  Philosophy}, page 108–165. Cambridge University Press.

\bibitem[Fraser, 2020]{Fr20}
Fraser, D. (2020).
\newblock {Particles in Quantum Field Theory}.
\newblock {\em Unpublished}.

\bibitem[Gradshteyn and Ryzhik, 2007]{GrRy07}
Gradshteyn, I. and Ryzhik, I. (2007).
\newblock {\em {Table of Integrals, Series, and Products}}.
\newblock Elsevier, New York, 7th edition.

\bibitem[Haag, 1992]{Ha92}
Haag, R. (1992).
\newblock {\em {Local Quantum Physics}}.
\newblock Springer, Berlin.

\bibitem[Halvorson, 2001]{Ha01a}
Halvorson, H. (2001).
\newblock {Reeh-Schlieder defeates Newton-Wigner: On alternative localization
  schemes in relativistic quantum field theory}.
\newblock {\em Philosophy of Science}, 68:111--133.

\bibitem[Halvorson and Clifton, 2002]{HaCl02}
Halvorson, H. and Clifton, R. (2002).
\newblock No place for particles in relativistic quantum theories?
\newblock {\em Philosophy of Science}, 69:1--28.

\bibitem[Hegerfeldt, 1998a]{He98}
Hegerfeldt, G. (1998a).
\newblock {Causality, particle localization and positivity of the energy}.
\newblock In B\"{o}hm, A., editor, {\em Irreversibility and Causality}, pages
  238--245. Springer.

\bibitem[Hegerfeldt, 1998b]{He98b}
Hegerfeldt, G. (1998b).
\newblock {Instantaneous spreading and Einstein causality in quantum theory}.
\newblock {\em Annalen der Physik}, 7:716--725.

\bibitem[Honegger and Rieckers, 2005]{HoRi05}
Honegger, R. and Rieckers, A. (2005).
\newblock {Some Continuous Field Quantizations, Equivalent to the C*-Weyl
  Quantization}.
\newblock {\em Publications of the Research Institute for Mathematical
  Sciences, Kyoto University}, 41(113-138).

\bibitem[Kay, 1979]{Ka79}
Kay, B. (1979).
\newblock {A uniqueness result in the Segal-Weinless approach to linear Bose
  fields}.
\newblock {\em Journal of Mathematical Physics}, 20:1712--3.

\bibitem[Kay, 1985]{Ka85}
Kay, B. (1985).
\newblock {The Double-Wedge Algebra for Quantum Fields on Schwarzschild and
  Minkowski Spacetimes}.
\newblock {\em Communications in Mathematical Physics}, 100:57--81.

\bibitem[Knight, 1961]{Kn61}
Knight, J.~M. (1961).
\newblock {Strict Localization in Quantum Field Theory}.
\newblock {\em Journal of Mathematical Physics}, 2:459--471.

\bibitem[Landsman, 1998]{La98b}
Landsman, N.~P. (1998).
\newblock {\em {Mathematical Topics Between Classical and Quantum Mechanics}}.
\newblock Springer, New York.

\bibitem[Landsman, 2013]{La13}
Landsman, N.~P. (2013).
\newblock {Spontaneous symmetry breaking in quantum systems: Emergence or
  reduction?}
\newblock {\em Studies in the History and Philosophy of Modern Physics},
  44:379--394.

\bibitem[Landsman, 2017]{La17}
Landsman, N.~P. (2017).
\newblock {\em {Foundations of Quantum Theory: From Classical Concepts to
  Operator Algebras}}.
\newblock Springer.

\bibitem[Letaw and Pfautsch, 1981]{LePf81}
Letaw, J. and Pfautsch, J. (1981).
\newblock {Quantized scalar field in the stationary coordinate systems of flat
  spacetime}.
\newblock {\em Physical Review D}, 24(6):1491--1498.

\bibitem[Licht, 1963]{Li63}
Licht, A. (1963).
\newblock {Strict Localization}.
\newblock {\em Journal of Mathematical Physics}, 4:1443--1447.

\bibitem[Madsen, 1988]{Ma88}
Madsen, M. (1988).
\newblock {Scalar fields in curved spacetimes}.
\newblock {\em Classical and Quantum Gravity}, 5:627--639.

\bibitem[Malament, 1996]{Ma96}
Malament, D. (1996).
\newblock In defense of dogma---why there cannot be a relativistic quantum
  mechanical theory of (localizable) particles.
\newblock In Clifton, R., editor, {\em Perspectives on Quantum Reality}.
  Kluwer.

\bibitem[Malament, 2012]{Ma12}
Malament, D. (2012).
\newblock {\em {Topics in the Foundations of General Relativity and Newtonian
  Gravitation Theory}}.
\newblock University of Chicago Press, Chicago.

\bibitem[Manuceau et~al., 1974]{MaSiTeVe74}
Manuceau, J., Sirugue, M., Testard, D., and Verbeure, A. (1974).
\newblock {The Smallest C*-algebra for the Canonical Commutation Relations}.
\newblock {\em Communications in Mathematical Physics}, 32:231--243.

\bibitem[Newton and Wigner, 1949]{NeWi49}
Newton, T. and Wigner, E. (1949).
\newblock {Localized States for Elementary Systems}.
\newblock {\em Reviews of Modern Physics}, 21:400--406.

\bibitem[Oldofredi and \"{O}ttinger, 2021]{OlOt21}
Oldofredi, A. and \"{O}ttinger, H. (2021).
\newblock The dissipative approach to quantum field theory: conceptual
  foundations and ontological implications.
\newblock {\em European Journal for Philosophy of Science}, 11(18).

\bibitem[Papageorgiou and Pye, 2019]{PaPy19}
Papageorgiou, M. and Pye, J. (2019).
\newblock {Impact of relativity on particle localizability and ground state
  entanglement}.
\newblock {\em Unpublished}, \texttt{arXiv:1902.10684v1}.

\bibitem[Peskin and Schroeder, 1995]{PeSc95}
Peskin, M.~E. and Schroeder, D.~V. (1995).
\newblock {\em {An Introduction to Quantum Field Theory}}.
\newblock Perseus Books.

\bibitem[Petz, 1990]{Pe90}
Petz, D. (1990).
\newblock {\em {An Invitation to the Algebra of Canonical Commutation
  Relations}}.
\newblock Leuven University Press, Leuven.

\bibitem[Redhead, 1995a]{Re95a}
Redhead, M. (1995a).
\newblock {More Ado About Nothing}.
\newblock {\em Foundations of Physics}, 25:123--137.

\bibitem[Redhead, 1995b]{Re95b}
Redhead, M. (1995b).
\newblock {The Vacuum in Relativistic Quantum Field Theory}.
\newblock {\em Philosophy of Science (Proceedings of the 1994 Biennial Meetings
  of the PSA)}, 2:77--87.

\bibitem[Reed and Simon, 1975]{ReSi75}
Reed, M. and Simon, B. (1975).
\newblock {\em {Methods of Modern Mathematical Physics II: Fourier Analysis,
  Self-Adjointness}}.
\newblock Academic Press, New York.

\bibitem[Reed and Simon, 1980]{ReSi80}
Reed, M. and Simon, B. (1980).
\newblock {\em {Functional Analysis}}.
\newblock Academic Press, New York.

\bibitem[Rieffel, 1989]{Ri89}
Rieffel, M. (1989).
\newblock {Deformation Quantization of Heisenberg manifolds}.
\newblock {\em Communications in Mathematical Physics}, 122:531--562.

\bibitem[Rieffel, 1993]{Ri93}
Rieffel, M. (1993).
\newblock {\em {Deformation quantization for actions of $\mathbb{R}^d$}}.
\newblock Memoirs of the American Mathematical Society. American Mathematical
  Society, Providence, RI.

\bibitem[Ruetsche, 2011]{Ru11}
Ruetsche, L. (2011).
\newblock {\em {Interpreting Quantum Theories}}.
\newblock Oxford University Press, New York.

\bibitem[Saunders, 1992]{Sa92}
Saunders, S. (1992).
\newblock {Locality, Complex Numbers, and Relativistic Quantum Theory}.
\newblock {\em Philosophy of Science (Proceedings of the 1992 Biennial Meetings
  of the PSA)}, 1:365--380.

\bibitem[Saunders, 1995]{Sa95}
Saunders, S. (1995).
\newblock {A Dissolution of the Problem of Locality}.
\newblock {\em Philosophy of Science (Proceedings of the 1994 Biennial Meetings
  of the PSA)}, 2:88--98.

\bibitem[Sebens, 2018]{Se18}
Sebens, C. (2018).
\newblock {Forces on fields}.
\newblock {\em Studies in the History and Philosophy of Modern Physics},
  63:1--11.

\bibitem[Sebens, 2019]{Se19}
Sebens, C. (2019).
\newblock Electromagnetism as quantum physics.
\newblock {\em Foundations of Physics}, 49(4):365--389.

\bibitem[Segal and Goodman, 1965]{SeGo65}
Segal, I. and Goodman, R. (1965).
\newblock {Anti-Locality of Certain Lorentz-Invariant Operators}.
\newblock {\em Journal of Mathematics and Mechanics}, 14(4):629--638.

\bibitem[Struyve, 2010]{St10}
Struyve, W. (2010).
\newblock Pilot-wave theory and quantum fields.
\newblock {\em Reports on Progress in Physics}, 73(10):106001.

\bibitem[Struyve, 2011]{Struyve11}
Struyve, W. (2011).
\newblock Pilot-wave approaches to quantum field theory.
\newblock {\em Journal of Physics: Conference Series}, 306:012047.

\bibitem[Valente, 2015]{Va15}
Valente, G. (2015).
\newblock {Restoring particle phenomenology}.
\newblock {\em Studies in the History and Philosophy of Modern Physics},
  51:97--103.

\bibitem[Wallace, 2001]{Wa01}
Wallace, D. (2001).
\newblock {Emergence of particles from bosonic quantum field theory}.
\newblock {\em Unpublished}, \texttt{arXiv:quant-ph/0112149v1}.

\bibitem[Wallace, 2006]{Wa06}
Wallace, D. (2006).
\newblock In defence of naivet\'{e}: The conceptual status of lagrangian
  quantum field theory.
\newblock {\em Synthese}, 151(1):33--80.

\bibitem[Weinberg, 1972]{We72}
Weinberg, S. (1972).
\newblock {\em Gravitation and Cosmology}.
\newblock Wiley \& Sons, Inc., New York.

\bibitem[Williams, 2018]{Wi17}
Williams, P. (2018).
\newblock {Scientific Realism Made Effective}.
\newblock {\em British Journal for Philosophy of Science}, 70(1):209--237.

\bibitem[Yakubovich, 1996]{Ya96}
Yakubovich, S.~B. (1996).
\newblock {\em {Index Transforms}}.
\newblock World Scientific, London.

\end{thebibliography}

\end{document}